\documentclass[12pt, English,titlepage]{article}
\usepackage{amsmath,amssymb,amscd,amsthm}
\usepackage{multicol,longtable,tabularx}
\usepackage{enumerate}
\usepackage{latexsym}
\usepackage{epsfig,graphicx}
\usepackage{natbib}

\newcommand{\PP}{\mathbb P}

\DeclareMathOperator{\Cov}{Cov} 
\DeclareMathOperator{\diag}{diag} 
 
\DeclareMathOperator{\PV}{PV}
\DeclareMathOperator{\LGD}{LGD}\DeclareMathOperator{\EXP}{EXP}
\DeclareMathOperator{\PTL}{PTL}

\newtheorem{theo}{Theorem}[section]

\newtheorem{lemma}[theo]{Lemma}

\theoremstyle{definition}

\newtheorem{remark}[theo]{Remark}
\numberwithin{equation}{section}
\numberwithin{figure}{section}
\numberwithin{table}{section}
\allowdisplaybreaks

\begin{document}

\title{Two Models of Stochastic Loss Given Default}
\author{Simone Farinelli\thanks{The opinions expressed in this document are our own and do not represent the ones of Kraus Partner Investment Solutions  AG.
        The risk control principles presented are not necessarily used by  Kraus Partner Investment Solutions  AG.
        This document does not provide a comprehensive description of concepts and methods used in risk control at  Kraus Partner Investment Solutions  AG.}\\
         Kraus Partner Investment Solutions  AG\\Milit\"arstrasse 76 \\ CH-8004 Zurich, Switzerland\\Email: farinelli@krauspartner.com
        \\\\
        Mykhaylo Shkolnikov \\Mathematical Sciences Research Institute \\17 Gauss Way \\Berkeley, CA 94720-5070 , USA\\Email: mshkolni@gmail.com
        }

\bigskip
\bigskip
\bigskip

\maketitle \thispagestyle{empty} \nopagebreak

\maketitle

\begin{abstract}
We propose two structural models for stochastic losses given default
which allow to model the credit losses of a portfolio of defaultable
financial instruments. The credit losses are integrated into a structural
model of default events accounting for correlations between the
default events and the associated losses. We show how the models can
be calibrated and analyze the impact of correlations between the
occurrences of defaults and recoveries by testing our models for a
representative sample portfolio.
\end{abstract}

\section{Introduction}
Many credit risk models assume that the losses given default (LGDs)
are a deterministic proportion of the exposures subject to
impairment and ignore the fact that LGDs can fluctuate according to
the economic cycle. For example \citet{ARS01}, \citet{ABRS05} show
that default rates and recovery rates are strongly negatively
correlated and measure a correlation of $0.75$ between yearly
average default rate and loss rates in the United States. They
provide strong correlation evidence between macro growth variables
(such as GDP) and recovery rates and test the impact of correlated
defaults and LGDs inferring an understatement of forecasted
portfolio losses by up to $30\%$. \par  \citet{GORV09} show strong
dependence of default rates and recovery rates on the economic cycle
for the time frame $2000$-$2009$ and detect strong negative
correlations between default and recovery rates in various U.S.
industries (e.g.: banking and finance, broadcasting and media)
between $2005$ and $2009$.\par An appropriate LGD model should have
a reasonable economic interpretation, and it should allow for a calibration by
available data and be based on a proper statistical setting. In
particular, the dependence structure between LGDs and default
indicators should not arise from a deterministic functional
relation.\par \citet{Fr00a} and \citet{Fr00b} propose a structural
model with a systematic risk factor representing the state of the
economy and driving both defaults and LGDs. The dependence of
default indicator and LGDs on the common risk factor gives rise to a
strong correlation between the two, which is in line with the
empirical evidence. Another single factor model has been proposed by
\citet{Ta04} and extended by \citet{Py03} who unifies Frye's and
Tasche's approaches.\par \citet{Hi06} introduces dependent LGD
modelling into a multi-factor latent variable framework providing a
good fit to corporate bond data. Marginal distributions for
indicator functions and LGDs can be specified.\par \citet{HKW07}
model default probabilities and LGDs jointly by generalized linear
mixed effect models with probit link function and inverse logit
function, respectively. All factors are observable, some represent
the general macroeconomical environment and others take obligors'
specificities into account.\par Inspired by the past research in
this field, the following manuscript proposes two models for
stochastic losses given default which are correlated across firms
and correlated with occurences of default events. Both models extend
a structural model for the default events to a joint structural
model for both defaults and LGDs. The described models have the
following key features:
\begin{itemize}
\item The LGDs are stochastic, correlated with each other and the occurences of default events.
\item The LGDs follow beta distributions with means estimated from historical data.
\item The shapes of the beta distributions vary across firms in such a way that the density function is concave if the corresponding credit instrument is backed by a collateral and convex otherwise.
\end{itemize}
The main differences between the two models are the following:
\begin{itemize}
\item In the first model the parameters of the LGD distributions are random depending on the expected LGD and the risk factors driving the losses in the case of default whereas in the second model the former are deterministic functions of the expected LGD.
\item In the first model the complete joint distribution of the LGDs and default indicators can be estimated whereas in the second model the correlations between the risk factors driving defaults and LGDs can be fitted exactly. At the same time the number of model parameters which have to be estimated coincides for the two models.
\end{itemize}
We believe that both models are statistically sound and have a
reasonable economic interpretation. Moreover, we provide  a
calibration methodology, which we apply to the available historical
data and test the models on a representative sample portfolio.

\section{A structural model for correlated defaults}

Following \cite{Pi04}, we present in this section a model for the
joint equity dynamics of counterparties appearing in a portfolio of
financial instruments. Since the default events are triggered by the
value of equity crossing a default barrier, this will lead naturally
to a structural model for the joint dynamics of defaults. The two
models for losses given default can be then viewed as attachments to
this model making a joint simulation of default events and losses
given default possible.

\subsection{Joint equity dynamics}

Let us consider $N$ firms and introduce the firm index $f=1,\dots, N$. Furthermore, we assume that these firms are spread over $I$ different industry categories $i=1,\dots,I$ and $R$ different regions $r=1,\dots,R$. We denote by $i_f$ and by $r_f$ the region and the industry category, respectively, which the firm $f$ belongs to. Firm $f$ belongs therefore to the industry-region cell denoted by
\begin{eqnarray*}
(i_f,r_f)\in\{1,\dots,I\}\times\{1,\dots,R\}=:\mathcal{IR}.
\end{eqnarray*}
We assume now that the equity process for all firms $E_t:=[E_t^1,\dots,E_t^N]^{\dagger}$ obeys in continuous time the
\textit{Geometric Brownian Motion} described by the SDE
\begin{eqnarray}
\frac{dE_t}{E_t}=\mu_t dt+\sigma_t dB_t+\tau_t dW_t,
\end{eqnarray}
\noindent whereby
\begin{itemize}
\item the processes $B_t:=[B_t^{1},\dots,B_t^{IR}]^{\dagger}$ and $W_t:=[W_t^1,\dots,W_t^N]^{\dagger}$ are two independent standard multivariate Brownian motions in $\mathbf{R}^{IR \times 1}$ and $\mathbf{R}^{N\times 1}$, respectively,
\item the functions $t\in[0,+\infty[\mapsto\mu_t:=[\mu_t^1,\dots,\mu_t^N]^{\dagger},     \sigma_t:=[\sigma_t^1,\dots,\sigma_t^N]^{\dagger}$ and $\tau_t:=[\tau_t^1,\dots,\tau_t^N]^{\dagger}$ are called     \textit{Drift}, \textit{Homoskedastic Volatility} and  \textit{Heteroskedastic Volatility}, respectively. More precisely, by homoskedasticity we mean that for all times $t$ and all firms $f$ we have        $\sigma_t^f=\sigma_t^{(i_f,r_f)}$, while in the heteroskedastic case the volatility components explicitly depend on the firm, that is $\tau_t=\tau_t^{f}$ \textit{cannot} be written as $\tau_t=\tau_t^{(i_f,r_f)}$,
\item products and fractions are unterstood componentwise.
\end{itemize}
\noindent The time discretization of the SDE leads to the panel model
\begin{eqnarray}
y_t=\alpha_t+\beta_t+\varepsilon_t
\end{eqnarray}
in which the following components in $\mathbf{R}^{N\times 1}$ represent
\begin{eqnarray*}
y_t:=\log E_t-\log E_{t-1}:&\quad\textit{ log returns } \text{for the firm equities,}\\
\alpha_t:=\mu_t-\frac{1}{2}(\sigma_t^{2}+\tau_t^{2}):&\quad\textit{deterministic firm specific components,}\\
\beta_t:=\sigma_t(B_t-B_{t-1})\thicksim\mathcal{N}\left(0,\diag(\sigma_t^{2})\right):&\quad\textit{region-industry systematic components,}\\
\varepsilon_t:=\tau_t(W_{t}-W_{t-1})\thicksim\mathcal{N}\left(0,\diag(\tau_t^{2})\right):&\quad\textit{firm specific idiosyncratic components.}
\end{eqnarray*}
\noindent The log-return of the asset value $y_t$ is therefore decomposed into a deterministic part $\alpha_t$, a stochastic homoskedastic part $\beta_t$ and an independent heteroskedastic part $\varepsilon_t$. Note that the multiplication and squaring operations applied to vectors are meant componentwise.\\\quad\\
We concentrate now on the systematic region-industry components which will be linked with the losses given default below. The former can be written as
\begin{eqnarray}
\beta_t=b_t\gamma_t+v_t
\end{eqnarray}
by setting
\begin{eqnarray}
\gamma_t\stackrel{d}{=}\mathcal{N}(0,I),\\
v_t\stackrel{d}{=}\mathcal{N}(0,\chi_t^2)
\end{eqnarray}
independently of each other and i.i.d. over time and letting $b^{(i,r)}_t\in\mathbf{R}$ be the industry-region ''beta'' coefficient. The process $\gamma^{r}_t$ can be interpreted as the region performance of log asset returns. The random variables in the components of the vector $\gamma_t$ can be defined as principal components of $\beta_t$ and $v_t$ is a residual quantity. The industry specific effects are homoskedastic within a region, that is they have the same variance ${(\chi^r_t)}^2$. Denoting by
\begin{eqnarray}
\rho^{r_1,\,r_2}_t:=\Cov_{t-1}^{\text{Stat}}(\gamma^{r_1}_t,\gamma^{r_2}_t)
\end{eqnarray}
the historical covariance available at time $t$, we set the covariance between the region-industry effects to the statistical covariance of the latter up to time $t-1$:
\begin{eqnarray}
\Cov_{t-1}^{\text{Stat}}\Big(\beta^{(i_1,\,r_1)}_t,\beta^{(i_2,\,r_2)}_t\Big)=\rho^{i_1,\,i_2}_tb^{(i_1,\,r_1)}_tb^{(i_2,\,r_2)}_t+\chi^{r_1}_t\chi^{r_2}_t\delta^{i_1,\,i_2}\delta^{r_1,\,r_2}.
\end{eqnarray}

\subsection{Joint default dynamics}

Assuming that the joint equity dynamics is normalized in such a way that a default of firm $f$ occurs if the value of equity goes below $1$ and denoting by $X_t^f$ the default indicator process of firm $f$, the marginal conditional default probabilities read
\begin{eqnarray*}
\mathbb{E}_t\left[\left.X^f_{t+1}\right|X^f_{t}=0\right]=\PP_t[\log E_{t+1}^f\le0|\log E^f_{t}>0]
=\PP_t[G^f_{t+1}\le g^f_t|\log E^f_{t}>0]=\Phi(g^f_t)
\end{eqnarray*}
\noindent where we have defined
\begin{eqnarray}
G^f_t=\frac{\beta^{(i_f,r_f)}_t+\varepsilon^f_t}{\sqrt{\Big(\sigma_t^{(i_f,r_f)}\Big)^2+ \Big(\tau_t^f\Big)^{2}}}
\thicksim\mathcal{N}(0,1),\\
g^f_t
=\frac{-\log E^f_t-\mu^f_{t+1}+\frac{1}{2}\Big(\Big(\sigma_{t+1}^{(i_f,r_f)}\Big)^{2}+\Big(\tau_{t+1}^f\Big)^2\Big)}
{\sqrt{\Big(\sigma_{t+1}^{(i_f,r_f)}\Big)^2+\Big(\tau_{t+1}^f\Big)^2}},\\
\Phi(x)=\frac{1}{\sqrt{2\pi}}\int_{-\infty}^x du\, e^{-\frac{u^2}{2}}.
\end{eqnarray}
The joint distribution of $\{G^f_t\}_{f,t}$ can be simulated sequentially as a multivariate Gaussian distribution with vanishing conditional expectation and conditional covariance matrix having ones on the diagonal and off-diagonal entries
\begin{eqnarray*}
\Cov_{t-1}^{\text{Stat}}(G^{f_1}_t,G^{f_2}_t)=\Cov_{t-1}^{\text{Stat}}\left(\frac{\beta^{(i_{f_1},r_{f_1})}_t}{\sqrt{\Big(\sigma_t^{(i_{f_1},r_{f_1})}\Big)^2+\Big(\tau_t^{f_1}\Big)^2}},\frac{\beta^{(i_{f_2},r_{f_2})}_t}{\sqrt{\Big(\sigma_t^{(i_{f_2},r_{f_2})}\Big)^2+\Big(\tau_t^{f_2}\Big)^2}}\right).
\end{eqnarray*}
Therefore, to simulate defaults, the default probabilities $p^f_t$ from the macroeconomical model are inverted to
$c^f_t=\Phi^{-1}(p^f_t)$ and the vector-valued random variable \{$G^f_t\}_{f,t}$ is simulated. A simulated default event for firm $f$ occurs at the first time for which $G^f_t$ falls below $c^f_t$.

\section{Market \& credit exposures and losses given default}\label{defexplgd}

We now present a formal definition for market \& credit exposures and losses given default. To this end, let us consider a financial instrument at time $t=0,1,\dots,T$ and assume that it is $a_t-$rated, for a rating $a_t\in\mathcal{J}=\{1,2,\dots,J\}$, which is valid for the interval $]t-1, t]$. The financial instrument, or product, is  identified at time $t$ with the stochastic discrete cashflow stream $\{C_s\}_{s \ge t}$ that it generates from time $t$ on. In particular, it accounts for future possible defaults or rating migrations, but \textit{not} for recovery streams. When we evaluate it to determine its market value at time $t$, we need to use the $a_t$-term structure of interest rates. Assuming for the moment that there is only one currency, and after having denoted discount
factors by $d^{a_t}=d^{a_t}_{t,s}$ for $s\ge t$ and short rates by $r^{a_t}=r^{a_t}_{t}$, we can write the \textit{Present Value} of the product at time $t$ as
\begin{eqnarray*}
\PV_t(C;d^{a_t})=\sum_{s\ge t}\mathbb{E}_t^*\left[\exp\left(-\int_t^sdu\,r^{a_t}_{u}\right)C_s\right]
=\sum_{s\ge t}d^{a_t}_{t,s}\mathbb{E}_t^{*,s}[C_s].
\end{eqnarray*}
Thereby, $\mathbb{E}_t^*$ denotes the risk neutral conditional
expectation and $\mathbb{E}_t^{*,s}$ for $s \ge t$ the $s$-forward
neutral conditional expectation. We know that risk neutral
measure(s) and forward neutral measure(s) exist by virtue of the
Fundamental Theorem of Asset Pricing (see \citet{Bj04}) Chapters
10.2, 10.3 and 24.4). The present value of the product is its
\textit{Theoretical Price} and represents an approximation of its
\textit{Market Value} in the real world market.\par If a default
occurs at time $t$ for the state of nature $\omega\in\Omega$, then
the product cashflows are annihilated, that is $C_s(\omega)=0$ for
all $s\ge t$, and there possibly exist \textit{recovery} cashflows
$\{C^{\,\text{Rec}}_s(\omega)\}_{s \ge t}$, which will mitigate the
loss. These must be evaluated with respect to the government term
structure, which is considered to be default risk free. Therefore,
an approximation for the market recovery value is the theoretical
price of the recovery cashflow stream:
\begin{eqnarray*}
\PV_t(C^{\,\text{Rec}};d^{\text{Gov}})
=\sum_{s\ge t}\mathbb{E}_t^*\left[\exp\left(-\int_t^sdu\,r^{\text{Gov}}_{u}\right)C^{\,\text{Rec}}_s\right]
=\sum_{s\ge t}d^{\text{Gov}}_{t,s}\mathbb{E}_t^{*,s}[C^{\,\text{Rec}}_s].
\end{eqnarray*}
\textit{Exposure at Market \& Credit Risk} is then defined as
\begin{eqnarray}
\EXP_t(C):=\PV_t(C;d^{a_t})=\sum_{s\ge t}d^{a_t}_{t,s}\mathbb{E}_t^{*,s}[C_s].
\end{eqnarray}
\noindent Remark that this exposure definition does not presume that the counterparty is in the default state at time $t$ or later, but covers all possible future states, the default and the non-default ones. It is just the present value of the cashflow stream for all possible future states. Now, we can formally define the \textit{Loss Given Default} at
time $t$ for the financial instrument as
\begin{eqnarray}
\LGD_t(C;C^{\,\text{Rec}}):=1-\frac{\PV_t(C^{\,\text{Rec}};d^{\text{Gov}})}{\PV_t(C;d^{a_t})}
=1-\frac{\sum_{s\ge
t}d^{\text{Gov}}_{t,s}\mathbb{E}_t^{*,s}[C^{\,\text{Rec}}_s]}{\sum_{s\ge
t}d^{a_t}_{t,s}\mathbb{E}_t^{*,s}[C_s]}.
\end{eqnarray}
This definition is consistent with the definition for the recovery
of market value (see \citet{La04} Chapter 5.7 and \citet{Sc03}
Chapter 6).

\begin{remark}[Industry Standards for Traditional Credit Risk Management]

The definitions for loss given default and exposure presented here differ from the \textit{industry credit risk standards}, where:
\begin{itemize}
\item Exposure is typically termed Exposure-at-Default and is understood as conditional expectation of the theoretical value \textit{in the case of default}.
\item Loss given default is typically understood as a fraction.
\item Exposures are always positive. Typically, for structured products this is enforced by taking only the positive part of the distribution of the theoretical values.
\item For a \textit{loan or a bond} the exposure is typically defined in \textit{nominal terms}, in other words, as the \textit{issued amount}.
\item For a \textit{lombard} credit the exposure does not require netting of collateral, which needs to be treated separately.
\end{itemize}
\noindent Industry standards have of course their fundament. They service the book value accounting perspective, which segregates credit portfolio profits (cash-in-flows) from losses (cash-out-flows). In the traditional credit risk management approach \textit{only} losses are considered and these in nominal terms. For a book of plain-vanilla loans, bonds or  mortgages these approach can be reconciled with a mark-to-market valuation by modelling the
cash-in-flows. But for structured products like those from the investment banking, the requirement that exposure must be positive impedes diversification. It leads to a conservative overestimation of credit risk, which for a risk manager is on one hand reassuring, but on the other annoying, since it binds more risk capital than effectively necessary. We believe that our definition is the appropriate one to provide a \textit{fair valuation} for all products and to allow for an aggregation of both market and credit risks. As a matter of fact, using the notation introduced above, the value of a portfolio of $N$ financial instruments at time $t$ can be written as
\[V_t=\sum_{j=1}^N\left(1-X_t^j\LGD_t(C^j;C^{\,j,\,\text{Rec}})\right)\EXP_t(C^j),\]
whereas $X_t^j$ denotes the default indicator at time $t$ for the $j$th financial instrument. This formula shows how the portfolio value both depends on market and credit risk factors. Exposures depends on market risk factors, default indicators depend on credit risk factors, and loss-given-defaults depend on both.
\end{remark}

\section{Losses given default models}

In this section we explain two models for losses given default which allow for a joint simulation of default events and losses given default. Hereby, not only the default events are correlated, but the correlations between default indicators and losses given default are included as well. For both models we first present the theoretic framework and then explain how the model can be calibrated and used for simulations of the credit loss of a portfolio.

\subsection{Model A}

\subsubsection{Theoretic framework}

We propose a stuctural model for losses given default which is connected to the default times model through correlations between the risk factors in the two models. Following the recent literature we make the assumption that the loss given default of the financial instrument $f$ in the industry-region cell $(i,r)$ at time $t$ follows a beta-distribution $Beta(\mu_t^f,\nu_t^f)$. The corresponding parameters are modelled as functions of the risk factors driving the losses occurring at defaults. These risk factors are hereby of two types:
\begin{itemize}
\item systematic risk factors $Y_t^{(i,r)}$ characteristic for a industry-region cell $(i,r)$ at time $t$,
\item macroeconomic risk factor $\overline{Y}_t$ varying with the economic cycle and common to all industry-region cells.
\end{itemize}
Moreover, we assume that the distribution of losses given default of counterparties of industry $i$ in region $r$ depends on the value of the combined risk factor
\begin{eqnarray}
Z_t^{(i,r)}=\eta_t^{(i,r)}Y_t^{(i,r)}+\sqrt{1-\Big(\eta_t^{(i,r)}\Big)^2}\overline{Y}_t
\end{eqnarray}
and work directly with the latter for the purposes of parameter estimation and simulation. The risk factors $Z_t^{(i,r)}$ are assumed to be jointly normally distributed with the credit risk factors specified in section 2 with the only non-trivial covariances being
\begin{eqnarray}
cov\Big(Z_t^{(i,r)},Z_t^{(i',r')}\Big)=\theta^{i,r,i',r'},\\
cov\Big(Z_t^{(i,r)},\beta_t^{(i',r')}\Big)=\psi^{i,r,i',r'}.
\end{eqnarray}
It remains to specify the dependence of the parameters $\mu_t^f$, $\nu_t^f$ of the beta distribution on the risk factor $Z_t^{(i,r)}$. We note first that since the mean of the beta distribution $Beta(\mu,\nu)$ is given by $\frac{\mu}{\mu+\nu}$, it suffices to specify $\mu_t^f$ as a function of $Z_t^{(i,r)}$, since the value of $\nu_t^f$ is then automatically determined after a value for the expected loss given default is prescribed. Since $Z_t^{(i,r)}$ reflects the regional and industry type specificities as well as the economic cycle, we are quite free in our choice of a functional dependence between the stochastic parameter $\mu_t^f$ and the stochastic factor $Z_t^{(i,r)}$, as long as this dependence is given by a bijective function. Later, when calibrating the model different choices of the function will lead to different covariance parameters $\theta$ and $\psi$. Since $Z_t^{(i,r)}$ is Gaussian and $\mu_t^f$ assumes always strictly positive values, we can choose
\begin{eqnarray}
\mu_t^f=e^{Z_t^{(i,r)}},
\end{eqnarray}
because the exponential is a bijective function from the real axis to the positive reals. The mean of the loss given default distribution $m_t^f$ is already given as the result of estimation from historical losses given default data for the industry-region cell $(i_f,r_f)$ or as the result of expert judgement based on information about the corresponding counterparty. Finally, to ensure that the mean of the loss given default distribution is matched, we set
\begin{eqnarray}
\nu_t^f=\mu_t^f\cdot\frac{1-m_t^f}{m_t^f}.
\end{eqnarray}

\subsubsection{Parameter estimation and simulations}

We assume that the model for default times is already implemented and the time series of the risk factors are  already estimated. In this section we propose a simple and fast to implement estimation procedure which allows the joint modelling of the default times and the losses given default according to the previous sections. The estimation procedure is composed of the following three steps:
\begin{enumerate}
\item For each industry-region cell $(i,r)$ at time $t$ we compute the maximum likelihood estimate of the parameter $\mu_t^f$ (taking the same value for all firms $f$ belonging to the same cell).
\item We obtain a time series for the combined loss given default risk factor for each industry-region cell $(i,r)$ up to time $t$ from the estimates of the parameters $\mu_t^f$ by the formula
\begin{eqnarray}
Z_t^{(i,r)}=\log \mu_t^f
\end{eqnarray}
\item Lastly, we use the time series for $\beta_t^{(i,r)}$ and $Z_t^{(i,r)}$ to obtain an estimate on the covariances
\begin{eqnarray*}
cov\Big(Z_t^{(i,r)},\beta_t^{(i',r')}\Big): && 1\leq i,i'\leq I,\;1\leq r,r'\leq R,\\
cov\Big(Z_t^{(i,r)},Z_t^{(i',r')}\Big): && 1\leq i,i'\leq I,\;1\leq r,r'\leq R.
\end{eqnarray*}
\end{enumerate}
Consecutively, the loss given default distributions can be modelled in the following two steps:
\begin{enumerate}
\item At each point of time $t$ model the loss given default risk factors $Z_t^{(i,r)}$ jointly with the credit risk factors $\beta_t^{(i,r)}$, $\varepsilon_t^f$ as a multivariate normal random variable with the previously estimated covariance structure.
\item Set the parameters of the loss given default distribution for the financial instrument $f$ in cell $(i,r)$ at time $t$ by the formulas
\begin{eqnarray}
\mu_t^f=e^{Z_t^{(i,r)}},\\
\nu_t^f=\mu_t^f\cdot\frac{1-m_t^f}{m_t^f}
\end{eqnarray}
where $m_t^f$ is again the estimated or prescribed expected loss given default.
\end{enumerate}

\subsection{Model B}

In this second model losses given default will be also modelled as random variables assuming that the expectation of the loss given default conditioned on a default event affecting the corresponding company
\begin{eqnarray}
\overline{\LGD_t^f}=\mathbb{E}[\LGD_t^f|X_t^f=1]
\end{eqnarray}
has been already determined by estimation or expert judgement. In
previous work (see \citet{Le08} and \citet{Sc03}) the losses given
default (conditional on a default event) are assumed to be
beta-distributed:
\begin{eqnarray}
\LGD_t^f|(X_t^f=1)\sim Beta(\mu_t^f,\nu_t^f)
\end{eqnarray}
where
\begin{eqnarray}
\mu_t^f = (\kappa-1)\overline{\LGD_t^f}\\
\nu_t^f = (\kappa-1)(1-\overline{\LGD_t^f}).
\end{eqnarray}
Note that for any choice of the shape parameter $\kappa$ the desired expectation is matched. For the special choice   $\kappa:=4$ it was believed - on the basis of a long literature list - that the density of the loss given default distribution has a concave shape. This turns out to be true only if the expected loss given default is near $50\%$. However, the density is no more concave if the latter is close to $100\%$ or to $0\%$. If the expected losses given default are close to $0\%$ or to $100\%$, then the corresponding probability density near $1$ and $0$ is not negligible.\par
We assume as before that the conditional expectation $\overline{\LGD_t^f}$ of the loss given a default event is our best guess for the average loss rate in the case of a default event. Moreover, we recall that in the case of a collateralized loan the empirical loss given default distribution tends to have a concave shape, putting most of the weight on a small interval around the expectation. This is of course explained by the non-trivial recovery values. For unsecured loans the empirical loss given default distribution tends to have a convex form accounting for frequent loss given default values close to $1$. Therefore, we propose to model the density as a symmetric function with respect to the expectation. To this end, we apply a linear transformation to a beta distributed random variable $B\sim Beta(\mu,\nu)$:
\begin{eqnarray}
\LGD_t^f|(X_t^f=1)=&(\overline{\LGD_t^f}-\delta_t^f)B +  (\overline{\LGD_t^f}+\delta_t^f)(1-B),
\end{eqnarray}
where
\begin{eqnarray}
\mu=\nu=2,\quad \delta_t^f=0.2\min(\overline{\LGD_t^f}, 1-\overline{\LGD_t^f})
\end{eqnarray}
in case that the financial instrument $f$ is backed by collateral and
\begin{eqnarray}
\mu=\nu=0.5,\quad \delta_t^f=\min(\overline{\LGD_t^f}, 1-\overline{\LGD_t^f})
\end{eqnarray}
otherwise. This choice guarantees the symmetry of the probability density function with an appropriate support in a symmetric interval $[\overline{\LGD_t^f}-\delta_t^f, \overline{\LGD_t^f}+\delta_t^f]$ around the expectation. This choice reflects our knowledge (and ignorance) about the loss-given-default: we have a best guess given by its expectation but no opinion about its skewness. Therefore, a symmetric distribution supported by a neighborhood of the expectation is a legitimate choice.\\\quad\\
To specify the joint distribution of default events and losses given default recall that the standardized company equity return conditional on default is truncated Gaussian distributed. More precisely,
\begin{eqnarray*}
G_t^f|(X_t^f=1)&=G_t^f|(G_t^f<\Phi^{-1}(p^f_t))\sim F_{G_t^f|(G_t^f<\Phi^{-1}(p^f_t))}
\end{eqnarray*}
where
\begin{eqnarray}
F_{G_t^f|(G_t^f<\Phi^{-1}(p^f_t))}(x)=\frac{\Phi(x)}{p^f_t}
\end{eqnarray}
for all $x<p^f_t$ and $1$ otherwise. Therefore, assuming that the simulation of standardized equity returns has already occurred, we simulate losses given default in such a manner that they follow the specified marginal distributions and that their correlations with the standardized equity returns match their historical values. More exactly, we set
\begin{eqnarray}
\LGD_t^f=F^{-1}_{\LGD_t^f|(X_t^f=1)}(F_{H^f_t}(H^f_t)),
\end{eqnarray}
where
\begin{eqnarray*}
H^f_t=F_{G_t^f|(G_t^f<\Phi^{-1}(p^f_t))}(G_t^f|(G_t^f<\Phi^{-1}(p^f_t)))\left(1+\xi^f_t\left(V_t^f-\frac{1}{2}\right)\right).
\end{eqnarray*}
Thereby,
\begin{itemize}
\item $F_{\LGD_t^f|(X_t^f=1)}$ is the cumulative probability distribution function of the marginal specified above,
\item $V_t^f$ is a $[0,1]$-valued uniformly distributed random variable, independent of all $G_t^f$,
\item $\xi^f_t$ is a parameter which can be chosen in such a manner that the covariance between $F_{\LGD_t^f|(X_t^f=1)}(\LGD_t^f|(X_t^f=1))$ and $F_{G_t^f|G_t^f<\Phi^{-1}(p^f_t)}(G_t^f|(G_t^f<\Phi^{-1}(p^f_t)))$ attains its historical value.
\end{itemize}
\noindent Thus, $H^f_t$ can be viewed as the standardized company equity return, conditioned on the occurence of a default event, transformed to the uniform distribution and perturbed by the auxilliary loss given default risk factor $V_t^f$. The distribution function of the random variable $H^f_t$ is given by the following lemma.



\begin{lemma}\label{Lemma1}
Let $\xi\in]2,\infty[$ and $U,V$ be two standard uniformly distributed independent random variables. Then,
\begin{eqnarray*}
H= U\left(1+\xi\left(V-\frac{1}{2}\right)\right)
\end{eqnarray*}
has support in $[1-\xi/2, 1+\xi/2]$ and distribution function
\begin{eqnarray*}
F_H(y)=
            \left\{
            \begin{array}{ll}
                \frac{1}{2}-\frac{1}{\xi}+\frac{y}{\xi}+\frac{y}{\xi}\cdot\Big(\log(\xi/2-1)-\log |y|\Big), & \hbox{$1 - \frac{\xi}{2} < y < 0$} \\
                \frac{1}{2}-\frac{1}{\xi}+\frac{y}{\xi}+\frac{y}{\xi}\cdot\Big(\log(1+\xi/2)-\log y\Big), & \hbox{$0 < y <  1 + \frac{\xi}{2}$.}
            \end{array}
            \right.
\end{eqnarray*}
\end{lemma}

\begin{proof}
Since $U$ and $V$ are standard uniformly distributed, their probability density functions are the indicator functions for the interval $[0,1]$. Furthermore, since $U$ and $V$ are independent, their joint probability density function is the indicator function of the square $[0,1]^2$. It follows
\begin{eqnarray}
F_H(y)&=P[H\le y]=\int_{\{u(1+\xi(v-\frac{1}{2})) \le y\}}du\,dv\,I_{[0,1]^2}(u,v).
\end{eqnarray}
\noindent Integration completes the proof.\\
\end{proof}

The choice of the parameters $\xi^f_t$ can be implemented by matching the historical covariance between $F_{\LGD_t^f|(X_t^f=1)}(\LGD_t^f|(X_t^f=1))$ and $F_{G_t^f|G_t^f<\Phi^{-1}(p^f_t)}(G_t^f|(G_t^f<\Phi^{-1}(p^f_t)))$. The equation connecting the two quantities is specified by the following lemma.

\begin{lemma}
The parameters $\xi^f_t$ fulfill

\begin{eqnarray}
\xi^f_t=\frac{10+8\sqrt{1+54\lambda^f_t}}{288\lambda^f_t-3},\quad\xi^f_t>2,
\end{eqnarray}
where
\begin{eqnarray*}
\lambda^f_t=\Cov\left(F_{G_t^f|(G_t^f<\Phi^{-1}(p^f_t))}(G_t^f|(G_t^f<\Phi^{-1}(p^f_t))),F_{\LGD_t^f|(X_t^f=1)}(\LGD_t^f|(X_t^f=1))\right).
\end{eqnarray*}
\end{lemma}

\begin{remark} [Calibration of Model B] The model can be calibrated by setting the value of $\lambda^f_t$ to the historical covariance of the appropriate quantities in the industry-region cell $(i_f,r_f)$ before time $t$, i.e.:
\begin{eqnarray*}
\Cov_{t-1}^{i_f,r_f,\text{Stat}}\left(F_{G_t^f|(G_t^f<\Phi^{-1}(p^f_t))}(G_t^f|(G_t^f<\Phi^{-1}(p^f_t))),
F_{\LGD_t^f|(X_t^f=1)}(\LGD_t^f|(X_t^f=1))\right).
\end{eqnarray*}
The parameters $\xi^f_t$ which then depend only on the particular industry-region cell $(i_f,r_f)$ can be subsequently obtained by using the equation in the lemma. Note that by the Cauchy-Schwarz inequality it holds
\begin{eqnarray}
\lambda^f_t\leq\frac{1}{12}=0.083=:\lambda_{max}
\end{eqnarray}
with equality iff the two random variables are perfectly correlated. Hence, we expect that the estimates for $\lambda^f_t$ lie in $[\frac{1}{7}\lambda_{max},\frac{2}{3}\lambda_{max}]$ in which case $\xi^f_t>2$ and the equation in the lemma can be applied. If the estimate falls in one of the intervals $[-\lambda_{max},\frac{1}{7}\lambda_{max}[$ or $]\frac{2}{3}\lambda_{max},\lambda_{max}]$, one should replace it by $\frac{1}{7}\lambda_{max}$ or $\frac{2}{3}\lambda_{max}$, respectively.
\end{remark}
\begin{proof}[Proof of Lemma 4.2]
For the ease of notation we drop all indices. Then, we have
\begin{eqnarray*}
F_{\LGD|(X=1)}(\LGD|(X=1))=F_H(H)=F_H\left(U\left(1+\xi\left(V-\frac{1}{2}\right)\right)\right),
\end{eqnarray*}
where we have introduced the random variable
\begin{eqnarray*}
U=F_{G|(G<\Phi^{-1}(p))}(G|(G<\Phi^{-1}(p)),
\end{eqnarray*}
so that $U$, $V$ satisfy the assumptions of Lemma \ref{Lemma1}. Therefore,
\begin{eqnarray*}
\lambda&=&\Cov(F_{\LGD|(X=1)}(\LGD|(X=1)),F_{G|(G<\Phi^{-1}(p))}(G|(G<\Phi^{-1}(p)))\\
       &=&\Cov(F_H(H),U)=\Cov\left(F_H\left(U\left(1+\xi\left(V-\frac{1}{2}\right)\right)\right),U\right)
\end{eqnarray*}
which can be explicitly computed, because the joint density of $U$ and $V$ is the indicator function of the unit square and $F_H$ is known from Lemma \ref{Lemma1}. The result of the computation is the expression for $\xi$ displayed in the lemma statement.\\
\end{proof}

\section{Impact analysis for a sample portfolio}
\subsection{Portfolio}
We consider an invented portfolio of approximatively $17000$ firms
distributed across $40$ rating classes, $9$ industry types and $14$
world regions. The rating class $1$ corresponds to the best possible
credit worthiness, while a firm displaying rating $40$ is in the
default state. The world regions and industry types considered are
shown in tables \ref{regions} and \ref{industries}, respectively.\par
To describe the portfolio we show the repartition of expected
potential losses at time $0$ for the $1$Y default horizon with
respect to rating (see Figure \ref{HistoRating}), industry (see Figure
\ref{HistoIndustry}) and region (see Figure \ref{HistoRegion}). The
potential loss at time $t$ for a firm with exposure $\EXP_t$ and
loss given default $\LGD_t$ is defined as $\PTL_t:=\LGD_t\cdot\EXP_t$.

\subsection{Calibration, simulation and numerical results}
To calibrate the parameters of the default model we utilized the
Moody's-KMV Asset Values database covering approximatively $30000$
companies. To calibrate the parameters of the losses given default
models we employed the Moody's-KMV recovery database from which we
extracted LGDs of approximatively $5000$ occurred defaults. We
considered the time range $2001$-$2006$ with monthly quotes.\par The
models for both defaults and LGDs presented in the preceding
chapters have been implemented  as a Monte-Carlo simulation in
Matlab. By construction the correlation between simulated LGDs and
default indicators corresponds to the historical values and LGDs are
beta-distributed. The resulting statistics for the portfolio yearly
credit loss for $500000$ simulations can be found in Figure
\ref{statistics}. Hereby, in the deterministic LGD model the stochastic
LGDs of Models A and B are replaced by their respective deterministic means.
The column EL refers to the expected loss of the portfolio, the columns
Q\_90, Q\_95, Q\_99, Q\_99.95 and Q\_99.98 display several value-at-risk
quantiles of the portfolio and the columns ETL\_90, ETL\_95, ETL\_99, ETL\_99.95
and ETL\_99.98 show the expected tail losses above the respective quantiles.
\par We remark that -at least for this example - there are little tangible differences between
the output of Model A and B on the one hand and deterministic LGD,
on the other. This effect is probably due to the dominance of good
rated companies in the sample portfolio. More differences in the
tails can probably be obtained by modelling  conditional
expectations of LGDs (conditional on default events) which strongly
vary with respect to the absolute expected LGD. In both Model A and
B the expected conditional LGDs are frozen to their absolute
expectations and stochasticity is induced by second order and higher
conditional moments differing from the absolute moments. This will
be a subject of future research.

\newpage
\begin{table}
\begin{tabular}{|c|l|}
  \hline
  Region & Countries\\\hline
  1 & Australia and New Zealand \\
  2 & Hong Kong, Korea, Malaysia,Singapore, Thailand and Taiwan \\
  3 & Rest of Asia \\
  4 & Japan \\
  5 & European Union \\
  6 & Switzerland  \\
  7 & Rest of Europe \\
  8 & United Kingdom \\
  9 & Offshore jurisdictions \\
  10 & Argentina, Brazil and Chile \\
  11 & Rest of Latin America \\
  12 & Unites States and Canada \\
  13 & Oil \\
  14 & Least developed countries \\
  \hline
\end{tabular}
\caption{World Regions}\label{regions}
\end{table}
\begin{table}
\begin{tabular}{|c|l|}
  \hline
  Industry & Activity\\\hline
  1 & Banks, Insurance Companies and other financials \\
  2 & Manufacturing including energy and mining \\
  3 & Services\\
  4 & Health Care \\
  5 & Real Estate \\
  6 & Technology  \\
  7 & Utilities \\
  8 & Government \\
  9 & Private Customers \\
  \hline
\end{tabular}
\caption{Industry Types }\label{industries}
\end{table}

\begin{figure}[h!]
\centering
\includegraphics[width = 15cm]{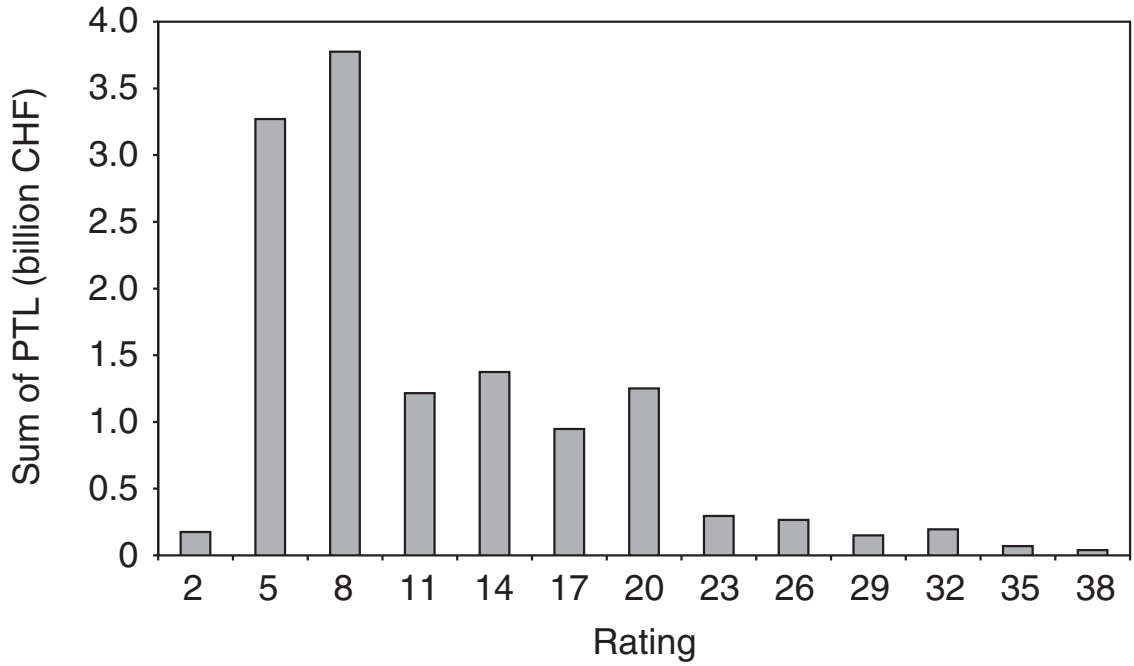}
\caption{Potential Loss vs. Rating (Billions CHF)}\label{HistoRating}
\end{figure}

\begin{figure}[h!]
\centering
\includegraphics[width = 15cm]{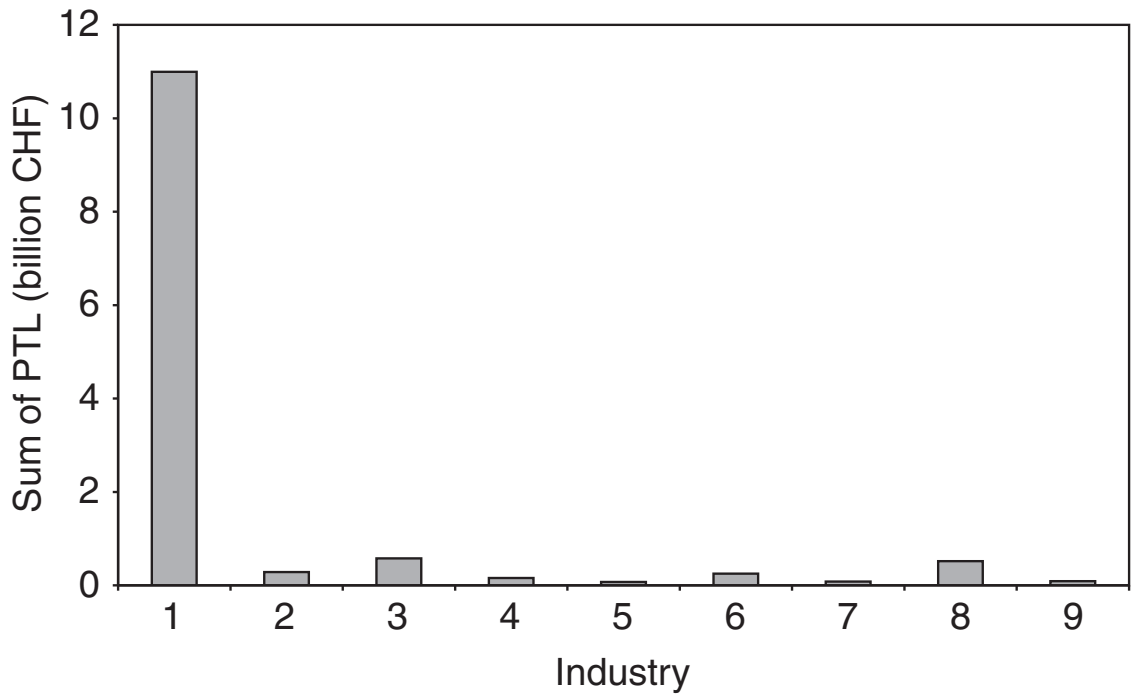}
\caption{Potential Loss vs. Industry (Billions CHF)}\label{HistoIndustry}
\end{figure}

\begin{figure}[h!]
\centering
\includegraphics[width = 15cm]{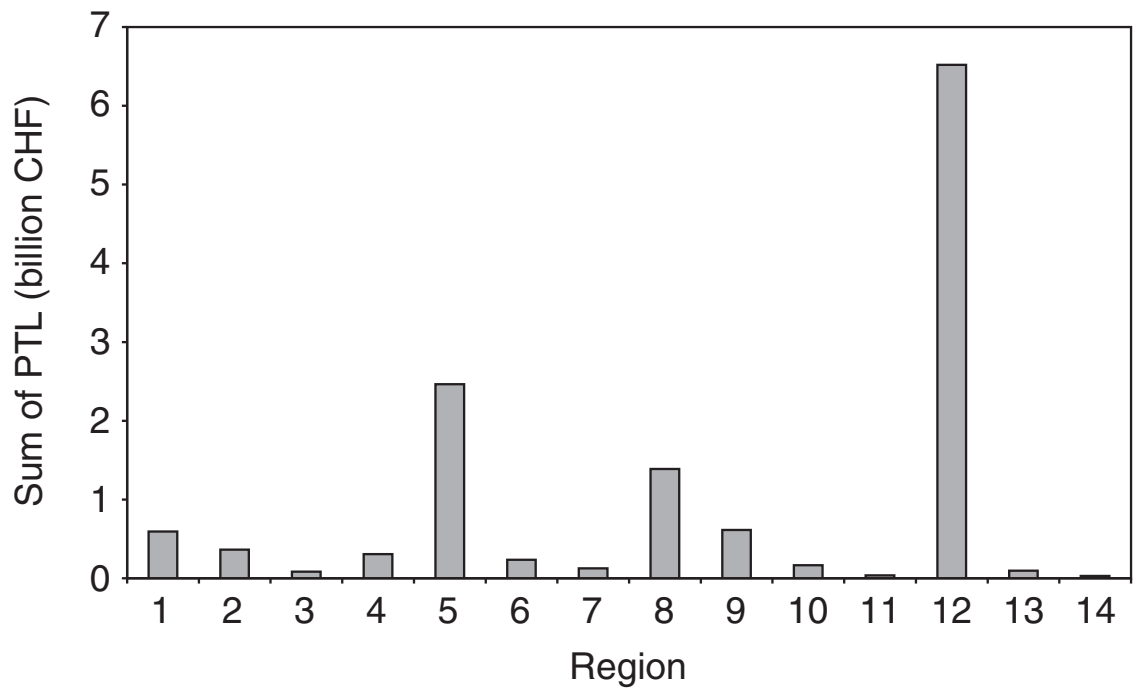}
\caption{Potential Loss vs. Region (Billions CHF)}\label{HistoRegion}
\end{figure}

\begin{table}
\begin{tabular}{|c|l|l|l|l|l|l|}
  \hline
  Model & $EL$ & $q_{90}$ & $q_{95}$ & $q_{99}$ & $q_{99.95}$ & $q_{99.98}$ \\
  \hline
  Deterministic LGD & 81.51 & 155.81 & 198.28 & 315.43 & 596.67 & 685.68 \\
  Model A & 81.46 & 155.72 & 198.64 & 315.91 & 596.58 & 678.50 \\
  Model B & 81.86 & 156.59 & 199.42 & 316.77 & 596.68 & 677.32 \\
  \hline
  \hline
  Model & & $\text{ETL}_{90}$ & $\text{ETL}_{95}$ & $\text{ETL}_{99}$ & $\text{ETL}_{99.95}$ & $\text{ETL}_{99.98}$\\
  \hline
  Deterministic LGD & & 223.47 & 272.34 & 402.54 & 686.08 & 769.06\\
  Model A & &223.43 & 272.23 & 401.17 & 678.56 & 753.33\\
  Model B & &224.19 & 273.11 & 402.64 & 678.80 & 753.78\\
  \hline
\end{tabular}
\caption{Loss Statistics (Millions of CHF)}\label{statistics}
\end{table}

\end{document}